\documentclass[12pt]{article}
\usepackage{amsmath,amsthm, amssymb, color, bbm, graphicx, tikz, color, setspace, subfigure, float}
\usepackage[utf8]{inputenc}

\usepackage{lscape, csquotes}
\usepackage{natbib} 
\usepackage{xcolor}
\usepackage[margin=1.0in]{geometry}
\usepackage[hidelinks]{hyperref}
\usepackage{multirow}
\usepackage{pdfpages}
\usepackage{titlesec}
\definecolor{background}{rgb}{0.84,0.92,0.95}
\newcommand{\btz}{\begin{tikzpicture}}
\newcommand{\etz}{\end{tikzpicture}}

\newcommand{\answer}[1]{\null}
\usepackage{caption}

\usepackage[utf8]{inputenc}
\DeclareGraphicsExtensions{.pdf,.ps,.png,.jpg,.emf,.eps}
\usepackage{longtable}
\usepackage{threeparttable}
\usepackage{enumitem}
\newtheorem{prop}{Proposition}
\newtheorem{prop_append}{Proposition}

\definecolor{darkblue}{rgb}{0.0,0.0,0.66}                 % Custom color: dark blue
\hypersetup{ 
    colorlinks=true,                                      % false: boxed links; true: colored links
    linkcolor=darkblue,                                   % Color of internal links
    citecolor=darkblue,                                   % Color of links to bibliography
    urlcolor=darkblue,                                    % Color of external links
}
\begin{document}
\title{Price dispersion across online platforms: Evidence from hotel room prices in London (UK)%
\thanks{The authors are grateful to Frank Verboven, Jan De Loecker, Jo Van Biesebroeck, Luis Vasconcelos, for invaluable help, advice, and encouragement. 
We sincerely thank the editor and the two anonymous referees for their detailed comments which led to significant improvements in the content and exposition of the paper. 
We also sincerely thank the seminar participants at KU Leuven and at the University of Essex for their valuable comments and helpful suggestions.
Lastly we thank Sunhyung Lee for helpful discussions.}
}
\author{Debashrita Mohapatra\thanks{Department of Agriculture and Resource Economics, University of Connecticut, Storrs, CT, USA E-mail: \href{mailto:debashrita.mohapatra@uconn.edu}{debashrita.mohapatra@uconn.edu }} \hspace{-2mm} 
\and Debi Prasad Mohapatra\thanks{Department of Resource Economics, University of Massachusetts Amherst, Amherst, MA, USA E-mail: \href{mailto:dmohapatra@umass.edu}{dmohapatra@umass.edu}} \hspace{-2mm} 
\and Ram Sewak Dubey\thanks{Economics Department, Montclair State University, Montclair, NJ, USA 07043; E-mail: \href{mailto:dubeyr@montclair.edu}{dubeyr@montclair.edu} (Corresponding Author)}}

\date{\today}
\maketitle

\begin{abstract}
{\noindent This paper studies the widespread price dispersion of homogeneous products across different online platforms, even when consumers can easily access price information from comparison websites. We collect data for the 200 most popular hotels in London (UK) and document that prices vary widely across booking sites while making reservations for a hotel room. Additionally, we find that prices listed across different platforms tend to converge as the booking date gets closer to the date of stay. 
However, the price dispersion persists until the date of stay, implying that the \enquote{law of one price} does not hold. We present a simple theoretical model to explain this and show that in the presence of aggregate demand uncertainty and capacity constraints, price dispersion could exist even when products are homogeneous, consumers are homogeneous, all agents have perfect information about the market structure, and consumers face no search costs to acquire information about the products. Our theoretical intuition and robust empirical evidence provide additional insights into price dispersion across online platforms in different institutional settings. Our study complements the existing literature that relies on consumer search costs to explain the price dispersion phenomenon. \\

\noindent\textbf{Key Words:} Price Dispersion, Online Platforms, Law of one price, Hotel Booking \\
\noindent\textbf{JEL Codes:} D40, D81, L83, L86.}
\end{abstract}

\newpage
%\onehalfspace
\section{Introduction}
Price dispersion is a widely observed phenomenon in offline and online markets, where sellers offer different prices for similar products.\footnote{See \citet{baye2006information}, \citet{baye2006persistent}, \citet{escobari2012dynamic}, \citet{brynjolfsson2000} for some of the evidences available in recent literature.} Traditional literature has explained the existence of price dispersion for homogeneous products through imperfect information among consumers and sellers (\cite{salop1982theory}, \cite{burdett1983equilibrium}). 
To quote \cite{sorensen2000equilibrium}, \enquote{\emph{Generally speaking, price dispersion will arise when there is a positive (but uncertain) probability that a randomly chosen customer knows only one price. Thus, even in markets with symmetric firms selling homogeneous products, prices may differ in equilibrium if consumers must incur search costs to obtain price information.}}

However, with the growth of technology and easy accessibility to internet services, gathering information about prices and product characteristics has become increasingly more accessible.  In the new age of the data revolution, the marginal costs of acquiring and transmitting information have consistently fallen. 
With a single click, consumers are now able to obtain information about prices charged by different sellers for products that range from computer hardware and software (\emph{Shopper.com}) to mortgages (\emph{Mortagequotes.com}) and airlines (\emph{farecompare.com}, \emph{kayak.com}). In the hotel industry, a consumer can access multiple platforms while booking a room. 
Different booking websites (e.g., Expedia.com, Booking.com, Hotels.com, etc.) post prices and availability for a given type of room at a given hotel for a specific date of stay. 
However, a consumer can also use a price comparison website (such as \emph{Skyscanner.com}, \emph{Trivago.com}), where she can check all the prices (and other room-related details) offered by different booking websites, compare them and book the hotel room using the platform that provides the lowest price. 
Effectively, the price comparison website works as an information exchange as it collates information from different booking websites and displays the combined information on prices (and other details for a hotel room) posted by different platforms as its search results. 
However, even though price comparison websites have been steadily increasing in popularity among consumers (\cite*{bodur2015online}), it is puzzling that price dispersion remains persistent across different markets selling homogeneous products.

In this paper, we focus on the hotel industry and investigate the dispersion of prices listed by booking platforms on a comparison website while booking a given type of room in a specific hotel for a given booking date. 
We provide robust empirical evidence for two findings: first, we document the existence of widespread price dispersion among homogeneous products in the hotel industry. 
In particular, we show that for a \enquote{given date of stay for a specific room type in a particular hotel}, prices across platforms vary widely for any booking date.  
Second, we focus on the evolution of price dispersion across booking platforms as the booking date is closer to the date of stay. 
We find that although the price dispersion persists, i.e., the law of one price does not hold, prices tend to converge (hence, dispersion goes down) as the booking date approaches the date of stay.

We do so by collecting detailed data from a price comparison website \emph{Skyscanner.com} for hotel rooms for the 200 most popular hotels in London, United Kingdom. 
To capture the existence of price dispersion across different platforms, we collect (through web scraping) prices posted for a specific date of stay across different booking platforms for a specified room type in a hotel. 
To ensure that the products are \emph{homogeneous}, we collect this information for a one-night stay (specified for two guests, 1 room)  on a given date of booking. 
We collected the price information for seven different dates of stay in November 2017 for all 200 popular hotels in London (UK). 
For each date of stay, we collected the price information on fifteen different booking dates.  
Therefore, each price information in our dataset corresponds to a \emph{date of stay--hotel id--room type} combination on a given website for a specific booking date. 

To document the existence of price dispersion, we compute the coefficient of variation of prices posted across booking sites for each \emph{date of stay--hotel id--room type} combination for a given booking date.%
\footnote{The coefficient of variation is given by the ratio of the standard deviation to the average price and captures the extent of variability about the average price.}  
Our results show that the coefficient of variation varies widely in our sample, providing strong evidence of price dispersion among booking websites for hotel rooms. 
On average, prices vary by 9\% around the mean hotel room prices. 
To further take care of any potential heterogeneity that may exist between booking websites (such as website reputation and user-friendliness of the booking site), we resort to regression analysis where we take into account the potential price variations due to time-invariant website characteristics, in addition to the hotel identity, room type, date and time of booking, and date of stay. 
We show that extensive price dispersion still exists across different booking websites even after controlling for a wide range of observable covariates. 

We then investigate how average price and price dispersion evolve as the date of booking approaches the date of stay. 
We find that the average hotel room price increases as we approach the date of stay. 
Additionally, as the booking date gets closer to the date of stay, prices listed across different platforms tend to converge, leading to a drop in the dispersion. 
Nevertheless, the price dispersion persists until the date of stay, implying that the law of one price does not hold.

Given the empirical evidence, we then investigate why price dispersion may exist for homogeneous products in our institutional setting. 
In particular, we ask: \enquote{in markets with symmetric firms selling homogeneous products, can prices differ in equilibrium if consumers incur zero search costs to obtain price information?}
We present a simplified model of the market to show that in the presence of aggregate demand uncertainty and capacity constraints, price dispersion could exist as an equilibrium outcome of the pricing game even when products are homogeneous, consumers are homogeneous, all agents have perfect information about the market structure, and consumers face no search costs to acquire information about the products. 
Uncertain demand is a common feature of the hotel booking industry, especially in big cities such as London, where the flow of customers to a specific hotel is often difficult to predict. 
Similarly, online booking platforms in the hotel industry face capacity constraints while making reservations as the hotels have only a limited number of rooms open for reservation. 
Additionally, a booking platform often has access to only a subset of rooms from a given hotel, adding to the capacity constraint within a booking website. 
Therefore, by showing the existence of price dispersion under such institutional settings but with zero search costs, our results provide additional insights into the existence of price dispersion. 
In doing so, we complement the existing literature that relies on consumer search costs to explain the price dispersion in equilibrium. 

Our study makes several contributions. 
First, we contribute by documenting widespread and persistent price dispersion among homogeneous products in the context of the hotel booking industry, where consumers can easily access price comparison websites. We show that although the law of one price does not hold, the prices tend to converge as we approach the date of stay. 
We highlight a unique setting, collect novel and detailed data from different hotel booking websites, and report robust evidence of price dispersion and convergence across different booking platforms for hotels in London. 
Second, our paper contributes to the growing literature on the study of price comparison websites by highlighting the role of such platforms in gathering information and potentially reducing consumers' search costs while making online purchases. 
Finally, we develop a simplified model to highlight the fundamental economic forces (in particular, demand uncertainty and capacity constraints) that may lead to price dispersion as an equilibrium outcome even when the consumers can acquire price information at zero search cost.

The rest of the paper is organized as under.
In  Section \ref{section:litreview}, we provide a review of the recent literature on price dispersion in online markets. 
Then we describe the details of the dataset collected from \emph{Skyscanner.com}, and document the empirical evidence of price dispersion in Section \ref{section:empiricalevidence}. 
Section \ref{section:model} describes the theoretical model, showing the existence of price dispersion even under zero search costs.
Section \ref{section:priceconvergence} documents the evidence of price convergence in the booking prices as we near the date of stay. We conclude in Section \ref{section:conclusion}.

\section{Literature review}\label{section:litreview}
Our paper is related to the extensive literature on online and offline markets that documents price dispersion for homogeneous products across sellers. 
The idea of \enquote{Law of one price} has been widely challenged in several theoretical and empirical studies in the existing literature. 
Starting with the seminal work by \cite{stigler1961economics}, several economists have developed various theoretical models based on information asymmetries, search frictions, price volatility, and demand uncertainty, among others, to explain this phenomenon and have documented the existence of price dispersion in offline markets (\cite{goldbergverboven}). 
With the increased adoption of online transactions, although consumers could access broader markets at relatively low cost, patterns of price dispersion among homogeneous products remain a widely reported phenomenon.
There are various reasons, however, to expect price dispersion to be lower on the internet platforms, as search costs are typically lower, and the websites offer significantly easier entry to the new sellers. 
With the advent of price comparison websites that collates information from different platforms and provides it to the consumers in one place, we expect the dispersion of prices for homogeneous product to reduce even further.

It is pertinent to discuss empirical evidence in some of the prominent recent papers. 
\citet{bailey1998} compared prices of 125 books, 108 music CDs, and 104 software titles in 1996 and 1997 through 52 internet and traditional outlets and found that price dispersion in online markets is as prevalent as in offline retailers. 
Also, he observed that prices on the internet are higher than on the offline channels. 
These findings contradict the theoretical predictions that expect a decline in price dispersion in a comparatively frictionless market on the internet.
Similarly, \citet{baylis2002} has documented that the pricing pattern observed in the market for Olympus C-2000Z digital camera (data for September 24- December 19, 1999, fourteen weeks period) and Hewlett-Packard 6300 flatbed scanner (data for October 7- December 19, 1999,  eleven weeks period) could be explained via price dispersion. 
\citet{brynjolfsson2000} compared prices for a matched set of 20 books and 20 CDs from 41 online and offline outlets between February 1998 and May 1999. 
They found that internet retailer prices differ by an average of $33\%$ for books and $25\%$ for CDs. 
However, by weighing the prices by proxies for market share, the dispersion is lower in internet channels than in conventional channels, reflecting the dominance of certain heavily branded retailers.
They figure out that, even if there is lower friction in many dimensions of internet competition, branding, awareness, and trust remain essential sources of heterogeneity among Internet retailers. 
\cite{orlov2011does} studied the influence of the internet on price dispersion by taking empirical evidence from the airline industry. 
This study shows that an increase in internet penetration reduces the average price and leads to higher intra-firm price dispersion. 
\citet{moller2016competition} did an empirical analysis to study the effect of price discrimination in airline markets and showed the positive effect of competition on price dispersion when demand uncertainty is high, or product differentiation is low. 
On the contrary, the effect is negative when preferences are sufficiently certain, and products are heterogeneous.%
\footnote{See \citet{dakshina2019}- a study of price dispersion over time and across geographical location, \citet{durba2014}- a study of price dispersion and inter-firm, inter-flight and frequency competition, \citet{xing2010}- a study of price dispersion for online branches of Multi-Channel Retailers and online-only retailers for recent references, \citet{choi2019}- a study of price dispersion for gasoline in 25 regions of Seoul, Korea in  response to asymmetric information between retailers and consumers.}

To explain those empirical patterns, existing literature has either relied on heterogeneity or search costs or both (\cite{arnold2011asymmetric}, \cite{sorensen2000equilibrium}, \cite{ellison2009search}, \cite{ellison2012search}, \cite{ellison2018search}, \cite{honka2019empirical}) to explain the existence of price dispersion in equilibrium. 
With the wider adoption of price comparison websites among consumers, several theoretical and empirical research has been done in the recent past to understand the price dispersion that persists even after the growth of the internet and the popularity of price comparison websites.
The study by \cite{baye2001information} focuses on the strategic incentive of a price comparison website known as \enquote{Information gatekeeper} and finds that it positively affects the consumers who enjoy free access to the sites. 
Access to information leads to competitive pricing and reduces the degree of price dispersion over time. 
A number of papers on clearing house models (for example, \citet{arnold2011asymmetric}; \citet{arnold2014unique}; \citet{baye2001information}; \citet{baye2004}; \citet{chioveanu2008advertising}; \citet{moraga2021consumer}; \citet{rosenthal1980model}; \citet{salop1977bargains}; \citet{varian1980model} have rationalized price dispersion in homogeneous goods markets as an equilibrium outcome through costly information acquisition leading to information asymmetry among consumers. 
Indeed, price dispersion has persisted despite technological advances such as the internet, extensive advertisement, and comparison sites. 
\citet{ronayne2015price} has investigated the impact of introducing a price comparison website to a market for a homogeneous good, where the sites charge firms for sale. 
They predict that a single price comparison website increases prices for all consumers, whereas multiple price comparison websites would lead to a profitable equilibrium. 
 The paper obtains price dispersion even in the presence of price comparison websites endogenously in equilibrium, which comes from the fact that some consumers have incomplete information.
Similarly, \citet{bodur2015online} reports the effect of price comparison websites on online and offline prices.

Our study complements this literature by empirically documenting widespread price dispersion in the hotel booking industry and provides a simple theoretical foundation where price dispersion can exist in equilibrium in markets with symmetric firms selling homogeneous products even when consumers incur zero search costs to obtain price information.

\section{Existence of Price Dispersion: empirical evidence}
\label{section:empiricalevidence}
For our analysis, we define a product as a ``date of stay–hotel id–room type'' combination.
Depending on when a consumer books a hotel room for a given date of stay, the price of a product might vary. For a specific date of stay, different booking platforms (e.g., Expedia.com, Booking.com, Hotels.com, etc.) post prices and availability for a given type of room in a given hotel. Therefore, prices can differ across platforms even for the same booking date. In our analysis, we focus on the price dispersion across platforms for a given date of booking. To document the price dispersion of products across different platforms, we collected detailed data from a price comparison website \emph{Skyscanner.com} for hotel rooms for the 200 most popular hotels in London, United Kingdom.\footnote{We used the web services provided by data extraction platform \emph{grepsr.com} to scrape the data for this study.}  
The price comparison website \emph{Skyscanner.com} works as an information exchange. 
For specific dates of stay, it collates information from different booking websites. 
It displays the combined price information (and other hotel room details) posted by different platforms as its search results. Hence, it provides the consumer with all relevant information in one place, saves the consumer  any search costs of visiting various websites separately, and helps her to find the room at the lowest available price. 

\begin{figure}[ht!]
\caption{Price Comparison Website} \label{fig:pricecompwebsite}
\centering
\includegraphics[trim=0in 0.5in 0in 0.5in,scale=0.68]{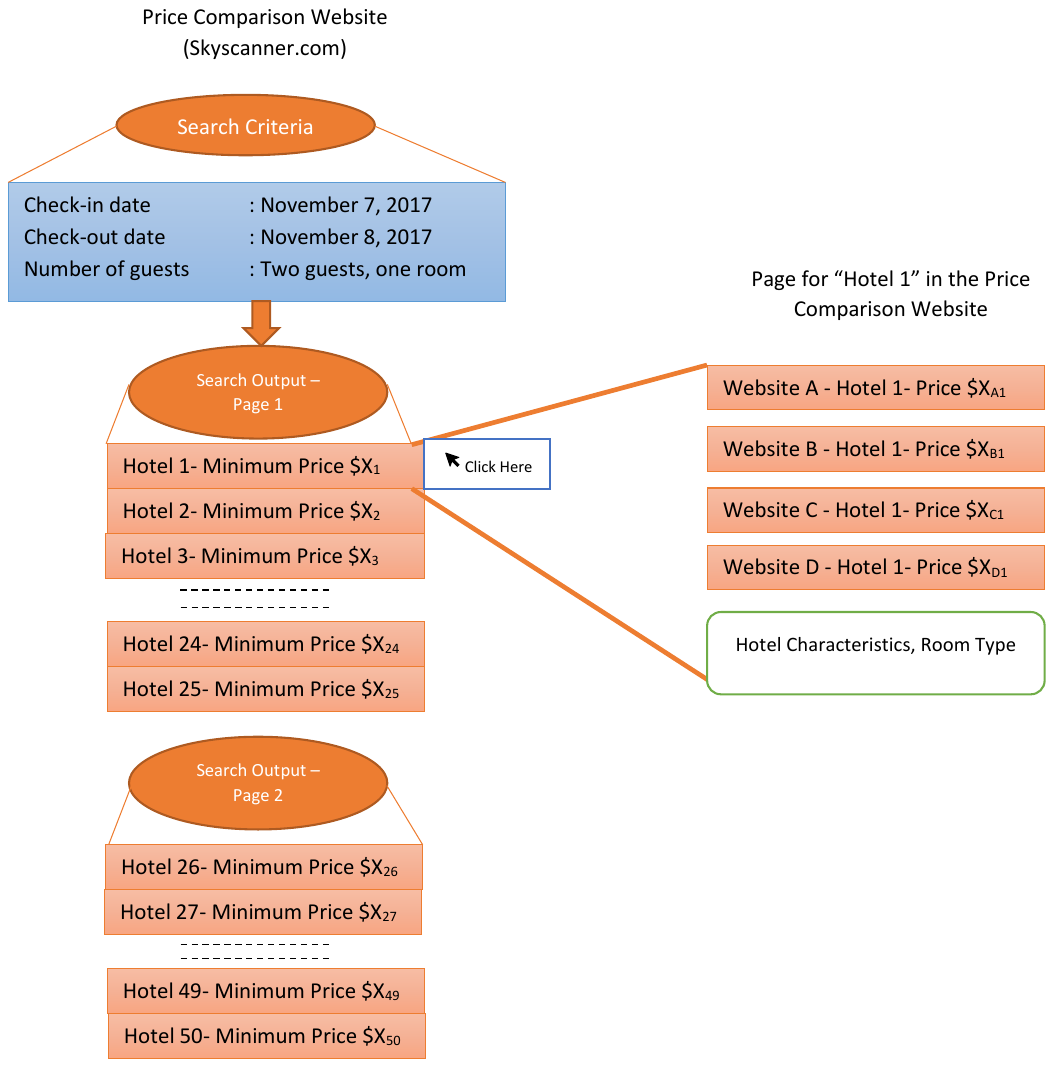}
\begin{tablenotes}
\footnotesize
\item[] Notes: 
This figure illustrates how data is organized on price comparison websites. On the left, one can see the search criteria and what a typical search displays on the Skyscanner website. Web page 1 presents hotel results from 1 to 25, web page 2 shows hotels 26 to 50, and so forth. If one clicks on Hotel 1, a new page opens with room prices for that hotel on various websites (shown on the right panel of Figure \ref{fig:pricecompwebsite}).
\end{tablenotes}
\end{figure}

Figure \ref{fig:pricecompwebsite} illustrates how data is organized on price comparison websites. As the left panel of the figure shows, when a user provides her travel details (like date of stay and number of guests), the price comparison website displays a list of hotels on different pages (also called web pages). Each web page typically lists around 20-25 hotels and their prices. During our sample period, a hotel search for London provided at most 14 such web pages. Clicking on a hotel link reveals its room prices on various websites. For instance, as the bottom left panel of figure \ref{fig:pricecompwebsite} shows, web page 1 lists prices for hotels numbered 1 to 25, and web page 2 lists prices for hotels 26 to 50. Clicking on Hotel 1
shows its prices on different websites (right panel of figure \ref{fig:pricecompwebsite}).

We take several steps while collecting the data to focus on the price dispersion of \emph{homogeneous} products across different platforms. 
First, we only consider prices for a one-night stay in a hotel room (specified for two guests, 1 room). 
The different dates of stay belong to the second week of November, such as November 7th check-in--November 8th check-out; November 8th check-in--November 9th check-out, and so on until November 13th check-in--November 14th check-out. 
Second, we collect detailed information about the website-hotel room-specific characteristics (whether the room is double or single).\footnote{A hotel might offer more room types (such as a suite or a family room) that may accommodate more than two guests. We did not collect information for those rooms in order to keep the set of products homogeneous.} %, whether the payment for the booking is refundable or non-refundable, and so on). 
Finally, for a given date of stay, we collect price data for all hotels on a specific date of booking at a particular time (for example, say for date-of-stay November 7th, we collect price data on October 24th, 2017, where October 24th is the date of booking). 
Day 1 of data collection starts on October 24th, 2017, at 10:30 am EST. 
On the first day of the data collection, we searched for all hotels to stay in London and sorted them by popularity.  
For our analysis, we considered the top 200 hotels ranked by popularity by the Skyscanner website on day 1. In other words, the 200 hotels in our sample were ranked as the 200 most popular hotels in London as selected by
Skyscanner’s algorithm on October 24th, 2017, at 10:30 a.m. EST, the first day of our data collection.\footnote{Skyscanner uses consumer reviews, hotel star ratings (2-star, 3-star, 4-star, and 5-star), and other details from booking and text reviews to rank the hotels by popularity.} We collected information for these 200 hotels in parallel every day at 10:30 a.m. EST until November 9, 2017. 

For each date of stay and the booking date, our search listed at most 14 web pages of results. We collected the information on the 200 hotels across the entire set of search results. Note that a hotel won’t appear in the search results if it is fully booked. If any of the 200 hotels do not appear on the search results, it implies that they’re booked for that date of stay. 
Our final dataset from London contains 179,234 observations. 
Each observation in our dataset corresponds to a price listed on a website for a \emph{date of stay--hotel id--room type} combination on a given booking date. 
To further verify whether the prices posted on the price comparison website are indeed the actual price to be paid by the consumer, we take a random sample of 1000 observations from the dataset. 
For each observation, we verify the price on the payment page for each posted website for the hotel room. 
We have confirmed that the prices posted on Skyscanner.com are the actual price to be paid if the consumer clicks on the corresponding website and makes the purchase.

\begin{table}[H]
\caption {Summary statistics for prices across websites for hotels in London} \label{summary_stats}
\begin{tabular}{l c c c c } \hline\hline\\[-3mm]
Variables  & Mean & Std.Dev. & Min & Max \\ \hline\\[-2mm]
Posted Price (in \pounds)\footnotemark[1] & 219 & 149 & 36 & 998  \\[2mm]
No. of websites\footnotemark[2] & 9 & 3 & 1 & 19 \\[2mm]
Range (in \pounds)\footnotemark[3]  & 55 & 75 & 0 & 986 \\[2mm]
Coefficient of Variation & 0.09 & 0.18 & 0 & 0.82 \\\hline
\multicolumn{5}{l}{\footnotesize \footnotemark[1] Posted Price refers to the price posted by different website} \\
\multicolumn{5}{l}{\footnotesize \footnotemark[2] No. of websites refers to the number of websites posting price for} \\
\multicolumn{5}{l}{\footnotesize  the same hotel room in skyscanner.com} \\
\multicolumn{5}{l}{\footnotesize \footnotemark[3]   Range refers to the
gap between the highest and lowest prices }\\
\multicolumn{5}{l}{\footnotesize   posted for a hotel room across websites}\\
\hline\hline
\end{tabular}
\end{table}
Table \ref{summary_stats} documents summary statistics from our data.
The average prices posted across websites varies between \pounds 36  and \pounds 998. 
This suggests that our dataset covers a wide range of hotels, both cheap, low-quality, and high-quality, expensive hotels. 
We observe that consumers can access posted prices from 9 different websites on average for each hotel room.
The number of websites for posted prices varies between 1 to 19.  
To capture the price dispersion on a given booking date across websites for a specific \enquote{date of stay--hotel id--room type} combination, we compute the price range as well as the coefficient of variation.   
For a specific date of stay, on a given date of booking, the price range refers to the difference between the maximum posted price and the minimum posted price across websites for a hotel room. 
Our results show on average, the maximum price exceeds the minimum price by \pounds 55. 

To compute the coefficient of variation on a given date of booking, we determine the average price across websites and the standard deviation of prices across websites for a \enquote{date of stay--hotel id--room type} combination. 
The coefficient of variation is given by the ratio of the standard deviation to the average price and captures the extent of variability in relation to the average price. 
If there is no price variation across websites, the standard deviation is 0, and hence the coefficient of variation is given by 0. 
A higher coefficient of variation reflects higher price dispersion across websites. 
As documented in table  \ref{summary_stats}, the average coefficient of variation is 0.09, suggesting that the relative variation of prices across websites is around 9\% of the average price.

\begin{figure}[H]
\caption{Scatter plot of Coefficient of Variation}
\centering
\includegraphics[trim=0in 0in 0in 0in,scale=1.0]{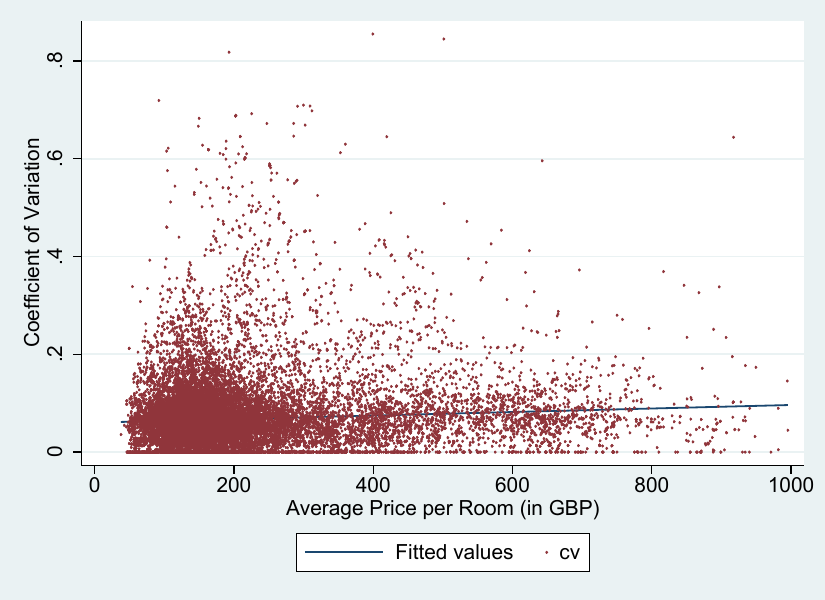}
%\floatfoot{
\newline
\tiny{Note: Here we have plotted the average prices for a \enquote{date of stay--date of booking--hotel id--room type} combination in the x-axis, and the  coefficient of variation of prices in the y-axis.}
%}
\label{figure:cv_london}
\end{figure}

As Figure \ref{figure:cv_london} shows, the coefficient of variation varies from 0 to about 0.8. This indicates significant price
differences among various websites when looking at a particular hotel room, date of stay, and booking date.
In appendix Section \ref{appendixsection:stdev}, we plot the standard deviation with respect to prices. We observe wide variation in the standard deviation values further providing evidence of price dispersion. 

Note that, although we have considered the hotel identity, the room type, date and time of booking, and date of stay, the differences in the websites displaying prices may lead to heterogeneity across products and may explain the price dispersion. 
To further address this concern, we run the following regression for each  \enquote{date of stay--date of booking} combination.
\begin{equation}\label{regression_equation1}
\begin{aligned}
\underset{d_s, d_b, h, r, w}{\text{Price}} = & \beta_0 + \beta_1 \left(\text{Hotel dummy}\times\text{Hotel room type dummy}\right)  \\
& + \beta_2 (\text{Website dummy}) + \varepsilon_{d_s, d_b, h, r, w} \\
\end{aligned}
\end{equation}
where $d_s$ stands for date of stay,  $d_b$ for date of booking, $h$ for hotel id, $r$ for room type,  and  $w$ for the website.
We run 106 separate regressions for each \enquote{date of stay--date of booking} combination.\footnote{Note that, for the date of stay Nov 7, 2017, we have data from 14 dates of booking prior to the date of stay. Similarly, for the date of stay Nov 13, 2017, we have data from 17 dates of booking prior to the date of stay. For all other 5 dates-of-stay, we have data from 15 days prior to the date of stay. Therefore, we have 106 regressions in total.} 
In each regression, we consider the prices of a specific room in a given hotel as we control for $\left(\text{Hotel dummy}\times\text{Hotel room type dummy}\right)$ in the regression. 
Additionally, we control for website dummies to control for any differences between the websites. For example, some websites may have a loyal consumer base, which may be the preferred website for that set of consumers. 
If the differences in websites can explain all price heterogeneity for a specific room in a given hotel, then given the regression specification, we should be able to explain the data perfectly. 
In particular, if website heterogeneity leads to price dispersion, controlling for website dummies should result in an $R^2$ close to 1 for most of the regressions. 
\begin{figure}[H]
\caption{Evidence of price dispersion across websites}
\centering
\includegraphics[trim=0.7in 2.5in 0.7in 2.5in,scale=0.5]{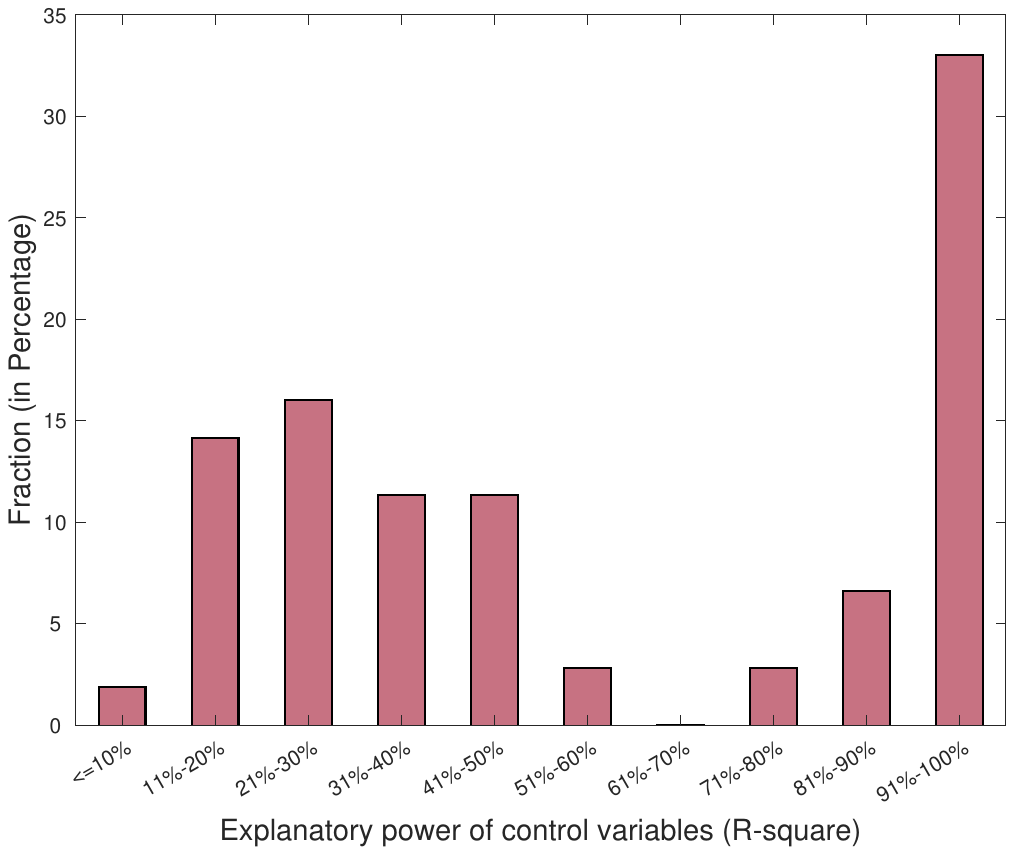}
%\floatfoot{
\newline
\tiny{Note: Here we have plotted the histogram of R-square for different \enquote{date of stay--date of booking} combinations as specified in equation \ref{regression_equation1}.}
%}
\label{figure:R-square}. 
\end{figure}
In figure \ref{figure:R-square}, we plot the histogram of $R^2$ values from 106 different regressions. 
Our results show that, only in 33\% of the cases, the $R^2$ is above 90\%. 
In 55\% of the cases, the estimated $R^2$ is less than 50\% suggesting that even controlling for heterogeneity across websites, for a given hotel room for a specific \enquote{date of stay--date of booking} combination, there exists substantial variation between prices posted across websites. 

Our empirical finding of price dispersion in the case of homogeneous products leads to the following question: \enquote{In markets with symmetric firms selling homogeneous products, can prices differ in equilibrium if consumers incur zero search costs to obtain price information?}. 
In the Section \ref{section:model}, we present a simplified model of the market to show that in the presence of \emph{aggregate demand uncertainty} and \emph{capacity constraints}, price dispersion could exist as an equilibrium outcome of the pricing game even when products are homogeneous, consumers are homogeneous, all agents have perfect information about the market structure, and consumers face \emph{zero} search costs to acquire information about the products. 
Our model complements the existing literature that relies on consumer search costs to explain the price dispersion in equilibrium.

\section{A Model to Explain Price Dispersion}
\label{section:model}
Can price dispersion exist for homogeneous products when all consumers can acquire information about the availability and prices of the products at zero search cost? Note that, in our setting, a price comparison website effectively acts as an information exchange; it collates all relevant information and provides it to consumers at zero search cost.  
In this section, we present a simplified model of the market to show the existence of price dispersion under this setting.
It is shown that price dispersion could exist even when products are homogeneous,  consumers are homogeneous, all agents have perfect information about the market structure, and consumers face no search costs to acquire information about the products.

The model specification follows the approach in \citet{arnold2011}.
There are two firms labeled $i = \{1,2\}$ producing identical goods with same marginal cost $c>0$ and zero fixed cost.
Each firm is an expected profit maximizer.
Thus firm $i$ earns a profit per unit of $P_i-c$ by selling the good at the price of $P_i$.
The production capacity of each firm is constrained by a maximum of one unit of the good.

The consumers are assumed to have identical reservation value $v$, with  $v>c$ to avoid trivial outcomes.
There is aggregate demand uncertainty (common to each firm) in this market with the following feature.
There are two possible realizations of aggregate demand - high and low. 
In the high demand state (which occurs with probability $0<\alpha<1$, there are two active consumers in the market, whereas, in the low demand state (with probability $1-\alpha$), there is one active consumer.
Both the probability of occurrence of each state and the demand quantity is common knowledge between the two firms.
The firm’s decision problem is to set the price $P_i$, which they are required to complete simultaneously and non-cooperatively before the realization of the high/low demand state.\footnote{In this model, the only factors contributing to price dispersion are capacity constraint and demand uncertainty.}
In this model specification, it is straightforward to prove that the equilibrium price set by each firm $P_i$ will be restricted to the interval $[c, v]$; the details are left for the reader as an easy exercise.
The proposition \ref{p1}  below shows that there exists no symmetric equilibrium in pure strategy.

\begin{prop}\label{p1} 
There is no symmetric pure strategy  Nash equilibrium in the market described above.
\end{prop}

\begin{proof}
Please see the proof of Proposition \ref{p1_appendix} in appendix section \ref{section:model_Appendix}.
\end{proof}

Proposition \ref{p1} establishes the existence of price dispersion in equilibrium under the presence of demand uncertainty and capacity constraints. 
Next, we will show that, relaxing either the assumption of demand uncertainty or the assumption of capacity constraint would support a pure strategy symmetric Nash equilibrium.

\begin{prop}\label{p2}
Under demand uncertainty and no capacity constraint, if consumers are homogeneous, there exists a symmetric pure strategy  Nash equilibrium.
\end{prop}
\begin{proof}
Please see the proof of Proposition \ref{p2_append} in appendix section \ref{section:model_Appendix}.
\end{proof}

In the next proposition, we show that symmetric equilibria exist (one for each state) when the firms are able to learn the state of the market demand (high or low) but continue to face capacity constraint.

\begin{prop}\label{p3}
Under no demand uncertainty, but with capacity constraint, there exists  a symmetric pure strategy  Nash equilibrium.
\end{prop}

\begin{proof}
Please see the proof of Proposition \ref{p3_append} in appendix section \ref{section:model_Appendix}.
\end{proof}

Propositions \ref{p2} and \ref{p3} demonstrate that relaxing either demand uncertainty or capacity constraint would support a pure strategy Nash equilibrium in the pricing game where both firms charge identical prices leading to no price dispersion. In proposition \ref{p4}, we characterize a 

\begin{prop}\label{p4} 
There exists a symmetric mixed strategy Nash equilibrium in the market described above.
\end{prop}
\begin{proof}
Please see the proof of Proposition \ref{p4_append} and characterization of the symmetric mixed strategy Nash Equilibrium in the appendix section \ref{section:model_Appendix}.    
\end{proof}
Our characterization of the symmetric mixed strategy Nash Equilibrium shows that each seller will charge a price in the interval specified below under the cdf as described here, i.e.,  
\begin{equation}
\label{profit_eq_firm18}
P_i \in [\alpha v+(1-\alpha)c, v], \; \text{and}\; F(P_i) = \frac{(P_i-c) - \alpha (v-c)}{(P_i-c)(1- \alpha)}, \; i=1, 2, 
\end{equation}
is the symmetric mixed strategy Nash equilibrium of this game.
Note that, as uncertainty is resolved where the high demand state occurs with probability one, i.e., $\alpha \rightarrow 1$, $P_i \rightarrow v$. 

\section{Evolution of Price Dispersion: empirical evidence}
\label{section:priceconvergence}
This section investigates how average price and price dispersion evolve as the date of booking approaches the date of stay. 
\begin{figure}[H]
\begin{center}
\caption{Mean Price up to the date of booking}
\label{mean_price}
\includegraphics[width=.6\textwidth]{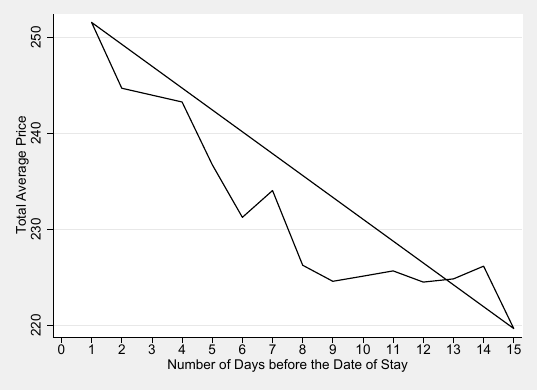}
\end{center}
\end{figure}

To track how the average price evolves as the booking date gets closer to the date of stay, we calculated the average price across all hotels, room types, websites, and stay dates based on a specific number of days before the stay date. This average price was then plotted against the number of days before the stay date in Figure \ref{mean_price}. Note that we gathered hotel price data for various stay dates and room types. For each date and room type, we collected prices from different websites for the 15 days leading up to the stay date. As we collected data for the 15 days preceding the stay date, there are 15 average values displayed on the graph.
 
As figure \ref{mean_price} indicates, the average hotel room price increases as we approach the check-in date (i.e., the x-axis values are closer to zero). 
This may result from hotels charging higher prices as they approach the capacity constraint. 
It may also indicate that price-insensitive business consumers book hotel rooms close to the booking date, and hence hotels may charge a higher price from those consumers effectively raising the average price. 
Next, we investigate if this implies higher (or lower) price dispersion while approaching the check-in date.

We conduct a set of regression analyses to provide evidence of price convergence across websites as we approach the date of stay. First, we run one-day ahead regression up to $14$ days before the date of stay and plot the coefficients and corresponding confidence intervals in figure \ref{cvtime_london}. The $k^{th}$ dot in the diagram refers to the coefficient of $CV_{t-k}$ in a regression where the dependent variable is $CV_{t-k+1}$ and the independent variable is $CV_{t-k}$. For example, the first dot in figure \ref{cvtime_london} plots the coefficient where we regress the  CV of prices for a given hotel on date $t$ on the CV of prices for the same hotel on date $(t-1)$. The second dot corresponds to the coefficient of $CV_{t-2}$ when the dependent variable is $CV_{t-1}$. 
\begin{figure}[H]
\begin{center}
\caption{Evolution of Price Dispersion (CV) up to the date of booking}
\label{cvtime_london}
\includegraphics[width=0.7\textwidth]{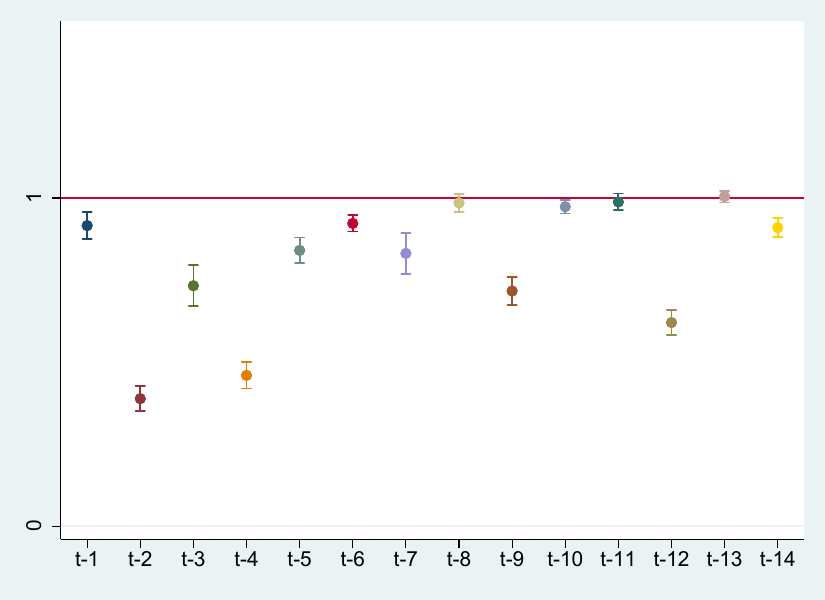}
\end{center}
\end{figure}

Our results from figure \ref{cvtime_london} suggests that all of the coefficient estimates are less than $1$ and almost all the confidence intervals of the coefficient estimates exclude $1$ with $95\%$ probability. 
This suggests that the CV on any day is less than the corresponding CV on the previous day as the coefficient is always less than $1$. 
Additionally, the figure suggests that the coefficient estimates for the last 6 days decline at an increasing rate. 
For example, moving from day $(t-6)$ to $(t-5)$, CV decreases by around $10\%$ where as moving from day $(t-5)$ to $(t-4)$, it decreases by more than $20\%$. 
Our supporting model shows that price dispersion exists in the equilibrium under demand uncertainty and capacity constraint. 
In addition to this, current literature suggests that consumer search costs may also contribute to price dispersion in equilibrium. 
Our results are consistent with both of these explanations. 
Our estimates show that in consistency with search theory, price dispersion persists over time, never converges to zero, and hence the law of one price does not hold (similar to \cite{baye2001information} and \cite{escobari2012dynamic}).\footnote{Note that if the law of one price holds on the date of stay, then $CV_{t}$ should be close to zero. This implies that the regression coefficient of $CV_{t}$ on $CV_{t-1}$ would be very close to zero. As figure \ref{cvtime_london} shows, the confidence interval for the regression coefficient of $CV_{t}$ on $CV_{t-1}$ is bounded away from zero and is close to 1, suggesting that price dispersion persists.} 
However, the decrease in the price dispersion over time is consistent with the hypothesis that since websites learn about the demand characteristics, the demand uncertainty gets resolved. 
This leads to a lower price dispersion as we get closer to the date of stay.

To provide further evidence, we regress the coefficient of variation (CV) at date $t$ for a given hotel on the corresponding CV for the same hotel from date $t-1$ while controlling for a set of hotel characteristics and several fixed effects. 
The regression specification is given by:
\begin{equation}
CV_t = \alpha CV_{t-1} + \beta_w X_w + \beta_h X_h + \epsilon    
\end{equation}
Our key parameter of interest is $\alpha$ which captures the evolution of CV after one day lag. 
$X_w$ denotes the website-specific characteristics, and $X_h$ denotes the hotel-specific characteristics. The controls include the number of websites that list the hotel price, the webpage number on which the hotel is displayed on Skyscanner, a dummy for the last three days before the date of stay,  number of reviews posted on websites for a hotel, hotel star rating (0
stars to 5 stars), and ratings by consumers (varies continuously from 1 to 10). We include fixed effects for hotels, date of stay, hotel quality, and ``hotel-date of stay-room type'' combinations in different specifications. 
Based on the ratings on the Skyscanner website, the hotels are grouped into different categories: Excellent, Very Good, Good, Average, and Satisfactory. The quality fixed effects
are based on this categorical variable. 
We report the results in table \ref{cvreg_london}. 
We present estimation results from seven different specifications. 
\begin{table}[H]
\begin{center}
\caption{Regression results - Coefficient of Variation}
\label{cvreg_london}
\scalebox{0.62}{
\begin{tabular}{lccccccc} \hline
 & (1) & (2) & (3) & (4) & (5) & (6) & (7) \\
Dependent variable: CV & Model1 & Model2 & Model3 & Model4 & Model5 & Model6 & Model7 \\ \hline
 &  &  &  &  &  &  \\
Previous day CV & 0.767*** & 0.758*** & 0.758*** & 0.757*** & 0.757*** & 0.507*** & 0.399*** \\
 & (0.00597) & (0.00598) & (0.00599) & (0.00600) & (0.00601) & (0.00845) & (0.00925) \\
Number of Websites &  & 0.0294*** & 0.0296*** & 0.0301*** & 0.0306*** & 0.0619*** & 0.0672***\\
 &  & (0.00249) & (0.00257) & (0.00259) & (0.00261) & (0.00378)  & (0.00406)\\
Page No in Skyscanner &  & -0.000851*** & -0.000809*** & -0.000764*** & -0.000798*** & -0.00239*** & -0.00356***\\
 &  & (0.000249) & (0.000255) & (0.000256) & (0.000257) & (0.000410) & (0.000604)\\
Dummy for last 3 days &  & 0.00475** & 0.00470** & 0.00475** & 0.00489** & 0.00497** & 0.00494** \\
 &  & (0.00201) & (0.00201) & (0.00201) & (0.00202) & (0.00197) & (0.00202) \\
No of Hotel Reviews &  &  & -0.000307 & -0.000112 & -0.000161 & 0.00499  \\
 &  &  & (0.00103) & (0.00105) & (0.00106) & (0.00940)  \\
Hotel star rating &  &  & -0.0113 & -0.0156 & -0.0151 & 0.0222  \\
 &  &  & (0.0108) & (0.0111) & (0.0111) & (0.358) \\
Hotel Review Rating &  &  & -0.00490 & -0.0116 & -0.00889 & -0.0158 \\
 &  &  & (0.0129) & (0.0329) & (0.0330) & (0.213) \\
Constant & 0.0171*** & -0.0305*** & -6.53e-05 & 0.00301 & -0.00434 & -0.127 & -0.0911\\
 & (0.00102) & (0.00536) & (0.0275) & (0.0595) & (0.0601) & (0.592)  & (0.151) \\
 &  &  &  &  &  &  \\ \hline
 &  &  &  &  &  &  \\ 
Observations & 10,829 & 10,829 & 10,829 & 10,829 & 10,829 & 10,829 & 10,829 \\
R-squared & 0.604 & 0.610 & 0.610 & 0.610 & 0.610 & 0.666 & 0.700\\ 
Hotel FE & No & No & No & No & No & Yes & No \\
Date of Stay FE & No & No & No & No & Yes & Yes & No \\
 Quality FE & No & No & No & Yes & Yes & Yes & No \\ 
  Hotel x Date Of Stay x Room Type FE & No & No & No & No & No & No & Yes \\ \hline
\multicolumn{7}{c}{ Standard errors in parentheses} \\
\multicolumn{7}{c}{ *** p$<$0.01, ** p$<$0.05, * p$<$0.1} \\ \hline
\end{tabular}}
\end{center}
\end{table}
Our results show that the estimated $\alpha$ is statistically significant irrespective of the model specification and lies between 0.40 and 0.77.\footnote{Note that in model 7, we include hotel id-date of stay-room type fixed
effects. Hence, we can not include controls such as the number of hotel reviews, hotel star
rating, and hotel review rating, as those variables are absorbed by these fixed effects.} 
Estimated $\alpha$'s value is less than 1 suggests that compared to the previous day, on average, CV goes down the next day as we approach the date of stay. 
The coefficient of the number of websites is positive and significant, in line with the hypothesis that a higher number of websites on which prices are posted leads to higher price dispersion. 
The coefficient of the Page number in Skyscanner on which a hotel is displayed is negative and significant. 
Hotels are sorted in terms of popularity, and a less popular hotel is displayed on a later page in our sample. 
Our estimates suggest that a less popular hotel faces lower demand and hence may have less price dispersion. 
The coefficient of the dummy for the last three days before the date of stay is positive and significant. 
This is consistent with the fact that business consumers and last-minute consumers show up in the last 3 days and are less price sensitive. 
Hence, price dispersion may increase due to the presence of such consumers. 
The coefficient estimates the number of hotel reviews, hotel star ratings, and ratings by consumers are not significantly different from 0. 
Given that CV is a standardized value and we do not use the variation across hotels, our results suggest that these hotel-specific characteristics may not significantly affect the evolution of CV. 
To provide robustness checks, we repeat the exercise for the logarithm of the coefficient of variation and price range. Further note that if prices are posted just on one website, then CV is equal to zero by definition. To make sure that such observations do not drive our results, we exclude those observations with a single website and run the regressions with the smaller sample. 
We report the results in the Appendix Section \ref{section:robustness_appendix}.

To summarize, our regression results highlight two consistent facts. 
First, price dispersion (measured in different ways) converges over time as the booking date approaches the date of stay. 
This is consistent with websites learning about uncertain market demand over time. 
Second, however, the law of one price does not hold, and price dispersion persists. 

\section{Conclusion}
\label{section:conclusion}
In this paper, we study the existence and persistence of price dispersion for homogeneous products across different online platforms in spite of the wide adoption of price comparison websites. 
We collect data from the 200 most popular hotels in London (UK) and document that the price to book a hotel room varies across booking websites. 
Our empirical evidence suggests that prices posted across websites may vary widely, even for a given \enquote{date of stay--date of booking--hotel id--room type} combination. 
Our regression analysis shows that such price dispersion patterns are robust to controlling for heterogeneity across the booking websites. 
We then investigate how the dispersion of prices evolves as the date of booking approaches the date of stay. 
We find that as the booking date gets closer to the date of stay, prices listed across different platforms converge, leading to a drop in the dispersion. 
Nevertheless, the price dispersion persists until the date of stay, implying that the law of one price does not hold.

We provide a simplified theoretical model of why price dispersion may exist for homogeneous products when all consumers can access an information exchange (such as a price comparison website) and acquire information about the availability and prices of the products at zero search cost. 
In contrast to the traditional literature that relies on consumer search costs to show the existence of price dispersion, we provide a model that shows that in the presence of aggregate demand uncertainty and capacity constraints, price dispersion could exist even when products are homogeneous, consumers are homogeneous, all agents have perfect information about the market structure, and consumers face no search costs to acquire information about the products. 

Our empirical findings show that in consistence with search theory, price dispersion persists over time, never converges to zero, and hence the law of one price does not hold. 
However, the decrease in the price dispersion over time is consistent with the hypothesis that since websites learn about the demand characteristics, the demand uncertainty gets resolved. 
This leads to a lower price dispersion as we get closer to the date of stay.  
Therefore, in the context of the hotel industry, our results complement the existing literature on search theory and provide additional insights into the existence and persistence of price dispersion phenomenon.

\clearpage
\newpage
\setlength{\bibsep}{-1pt}
\bibliographystyle{econ-aea} 
\small{\bibliography{reference}}

\clearpage
\newpage
\appendix
\begin{center}
    \Large\textbf{Appendix: For On-line Publication Only}
\end{center}
\section{Proofs}
\label{section:model_Appendix}

\begin{prop_append}\label{p1_appendix} 
There is no symmetric pure strategy  Nash equilibrium in the market described above.
\end{prop_append}

\begin{proof}
Observe that the expected profit of firm $1$, conditional on firm 2 setting a price $P_2$, when it sets its price equal to $P_1$, is given by,
\begin{equation}
\label{profit_eq_firm1}
\mathbbm{E}\pi_1 = (P_1 - c)\Pr(q_1(P_1, P_2)), 
\end{equation}
where $q_1(P_1, P_2)$ is the quantity sold by firm $1$ when the price vector set by the two firms is $(P1, P_2)$.
We need to consider two cases, (i) $P_1 = P_2 = P \in (c, v]$, and (ii) $P_1 = P_2 = P = c$.
\begin{enumerate}
\item[(i)]{$P_1 = P_2 = P \in (c, v]$: 
We show there exists a profitable deviation.
The expected payoff for firm 2 is
\begin{equation}\label{expected_pi_at_vv}
\mathbbm{E}\pi_2 = \alpha(P-c)+(1-\alpha)\left[t(P-c)+(1-t)(0)\right] = [\alpha + t(1-\alpha)](P-c), 
\end{equation} 
where we have assumed that in the low demand state firm 2 is able to sell one unit of good with probability $t\in (0, 1)$.%
\footnote{In the high demand state since the aggregate quantity demanded is $q=2$, each firm is able to sell one unit of good for sure due to capacity constraint.
On the other hand, in the low demand state  since the aggregate quantity demanded is $q=1$, firm 1 is able to sell one unit of good with a probability  $t\in (0, 1)$ and the firm 2 sells one unit with probability $1-t$.}
We next consider firm 2 deviating from the symmetric price $P$ to $P-\varepsilon \in (c, v)$, $\varepsilon>0$.
The expected profit at new price is 
\begin{equation}
\mathbbm{E}\pi_2^{\prime} = \alpha(P-c-\varepsilon) +(1-\alpha)(P-c-\varepsilon) = P-c-\varepsilon.
\end{equation}
Then,
\begin{equation}
\mathbbm{E}\pi_2^{\prime} - \mathbbm{E}\pi_2 = (P-c-\varepsilon)- [\alpha + t(1-\alpha)](P-c) = (1-\alpha)(1-t)(P-c) - \varepsilon >0       
\end{equation}
for sufficiently small deviation (take for example $0< \varepsilon < (1-\alpha)(1-t)(P-c)$).
Due to symmetric firm assumption, the argument goes analogously for Firm 1 where Firm 2 sets price at $P$.}
\item[(ii)]{$P_1 = P_2 = c$:
Note that the expected profit for Firm $2$ at  $P=c$ is $\mathbbm{E}\pi_2=0$. 
If Firm $2$ sets a price $c+\varepsilon\in (c, v)$ with $\varepsilon > 0$, then the expected profit is 
\begin{equation}
\mathbbm{E}\pi_2^{\prime}= \alpha(c+\varepsilon-c) + (1-\alpha) 0 = \alpha \varepsilon >0
\end{equation}
Note that Firm 2 sells one unit of the good in high demand state due to the binding capacity constraint of Firm 1.
Hence charging a price  ($c+\epsilon$) is a profitable deviation from setting price at $c$.} 
\end{enumerate}
 Due to symmetric firm assumption, the argument goes analogously for Firm 1 where Firm 2 sets price at $P$.
\end{proof}

\begin{prop_append}\label{p2_append}
Under demand uncertainty and no capacity constraint, if consumers are homogeneous, there exists a symmetric pure strategy  Nash equilibrium.
\end{prop_append}
\begin{proof}
In the absence of any capacity constraint, each firm can can cater to the aggregate demand in both high and low demand states. 

We rule out $P_1=  P_2 = P \in (c, v]$ as a symmetric equilibrium outcome.
For contradiction, we consider both firm charge price $P\in (c, v]$ and compute the profit earned by Firm 2.
In the high demand state, Firm 2 can sell two units with probability $0<r<1$ and one unit with probability $0<s<1-r$. 
In the low demand state, Firm 2 can sell one unit with probability $0<t<1$. 
If the firm deviates to a price $P-\varepsilon \in (c, v)$, it sells two units in the high demand state and one unit in the low demand state.
At price $P$, the profit is
\begin{equation}
\mathbbm{E} \pi_2= [\alpha(2r+s)+(1-\alpha)t](P-c).
\end{equation}
At price $P-\varepsilon$, the profit is
\begin{equation}
\mathbbm{E} \pi^{\prime}_2= [\alpha(2)+(1-\alpha)1](P-c-\varepsilon) = (1+\alpha)(P-c-\varepsilon).
\end{equation}
There exists a profitable deviation if $\mathbbm{E}\pi^{\prime}_2 > \mathbbm{E} \pi_2$ for some $\varepsilon>0$, i.e., 
\begin{equation}
\mathbbm{E} \pi^{\prime}_2- \mathbbm{E}\pi_2 = (1+\alpha)(P-c-\varepsilon)- [\alpha(2r+s)+(1-\alpha)t](P-c)>0.
\end{equation}
Simplifying, we get
\begin{equation}
[1+\alpha- \alpha(2r+s)-(1-\alpha)t](P-c)> \varepsilon (1+\alpha).
\end{equation}
Or
\begin{equation}
\left[(1-t)(1-\alpha r)+ \alpha(1-r-s)+\alpha t(1-r)\right](P-c)> \varepsilon (1+\alpha).
\end{equation}
Each term in the parenthesis on the left hand side is positive. 
Thus a profitable deviation $\varepsilon>0$ always exists.

Next, we claim that $P_1 = P_2 = c$ is a pure strategy symmetric equilibrium.
Since firms offer homogeneous products, at any demand state, if firm $2$ sets price $P=c$, it earns zero profit which would continue to be the level of profit if deviates to a price $P_i>c$.%
\footnote{Firm 2 ends up with zero sales as Firm 1 does not face any capacity constraint and would be able to sell two units in the high demand state and will sell one unit in low demand state for sure due to its lower price.}
Thus there is no profitable deviation and both firms charging same price equal to $c$ is the \emph{unique} symmetric pure strategy  Nash equilibrium. 
\end{proof}

\begin{prop_append}\label{p3_append}
Under no demand uncertainty, but with capacity constraint, there exists  a symmetric pure strategy  Nash equilibrium.
\end{prop_append}

\begin{proof}
In this setting, since we retain the assumption of capacity constraint, each firms has one homogeneous product to sell. 
Note that, there is no demand uncertainty, therefore, the firms are aware of the demand state while setting the prices. 
We claim that $P_1 = P_2 = v$ is the \emph{unique} pure strategy symmetric equilibrium when demand is high, and $P_1 = P_2 = c$ is the \emph{unique} pure strategy symmetric equilibrium when demand is low.

When aggregate demand is high, given that, firm $1$ sets price at $v$, firm $2$'s payoff from setting a price $P_2\in [c, v]$ is given by
\begin{equation}
(P_2 - c)\text{Pr}(q_2(v, P_2)).
\end{equation}
Since $\text{Pr}(\text{sale}(v, P_2))$ is equal to 1 under high state for any $P_2\in [c,v]$, setting $P_2 = v$ maximizes the payoff of firm 2. Hence \{$v,v$\} is a pure strategy symmetric equilibrium. 

At low demand state, when there are two firms and one consumer, given that firm $1$ sets price at $c$, firm $2$ sets the same price at equilibrium giving $0$ expected payoff to the firm. 
There is no profitable deviation for firm $2$ to charge more than $c$ which will lead to zero sales with probability one. 
Hence $P_1=P_c=c$ is a pure strategy symmetric equilibrium in this case, which is the equilirbium outcome in the usual Bertrand price competition market game. 
\end{proof}

\begin{prop_append}\label{p4_append} 
There exists a symmetric mixed strategy Nash equilibrium in the market described above.
\end{prop_append}

\begin{proof}
For any symmetric mixed strategy equilibrium in this game, the price support $[\underline{P}, \overline{P}]\subset[c,v]$ and the cdf $F(P_i)$ is a continuous function (i.e., it is atomless) for each of the two firms.
Note that for price $P_1$ of firm $1$, $P_2>P_1$ with probability $1-F(P_1)$.
Thus the profit for the firm 1 is 
\begin{equation}
\label{profit_eq_firm10}
\pi_1 (P_1)= (P_1 - c)\left(1-F(P_1)\right)\left[\alpha(1)+(1-\alpha)(1)\right] + (P_1-c)F(P_1)\left[\alpha(1)+(1-\alpha)(0)\right], 
\end{equation}
which is simplified to 
\begin{equation}
\label{profit_eq_firm11}
\pi_1 (P_1)= (P_1 - c)\left(1-F(P_1)\right) + \alpha (P_1-c)F(P_1). 
\end{equation}
Since $F(v)=1$, we can deteremine the profit for firm 1 at the mixed strategy equilibrium as 
\begin{equation}
\label{profit_eq_firm12}
\pi^{\ast}_1= \pi_1(v) = (v - c)\left(1-F(v)\right) + \alpha (v-c)F(v) = (v - c)\left(1-1\right) + \alpha (v-c)(1)= \alpha (v-c). 
\end{equation}
Observe that at every mixed strategy $P_1$, the profit of firm $1$, is given by,
\begin{equation}
\label{profit_eq_firm13}
\pi_1(P_1) = (P_1 - c)\left(1-F(P_1)\right) + \alpha (P_1-c)F(P_1) = \alpha (v-c). 
\end{equation}
This yields the cdf as 
\begin{equation}
\label{profit_eq_firm14}
F(P_1)  = \frac{(P_1-c) - \alpha (v-c)}{(P_1-c)(1- \alpha)}. 
\end{equation}
Given that the equilibrium profit for firm 1 is $\alpha (v-c)$, arbitrarily low prices could not be chosen. 
In particular, $P_1 \not< \underline{P}$ with $\underline{P}>c$ where $\underline{P}$ is such that $F(\underline{P})=0$.
\begin{equation}
\label{profit_eq_firm15}
F(\underline{P})  = \frac{(\underline{P}-c) - \alpha (v-c)}{(\underline{P}-c)(1- \alpha)}=0. 
\end{equation}
In other words,
\begin{equation}
\label{profit_eq_firm16}
\underline{P}  = c+\alpha (v-c) = \alpha v+(1-\alpha)c. 
\end{equation}
and 
\begin{equation}
\label{profit_eq_firm17}
\overline{P}  = v. 
\end{equation}
Therefore,
\begin{equation}
\label{profit_eq_firm19}
P_i \in [\alpha v+(1-\alpha)c, v], \; \text{and}\; F(P_i) = \frac{(P_i-c) - \alpha (v-c)}{(P_i-c)(1- \alpha)}, \; i=1, 2, 
\end{equation}
is the symmetric mixed strategy Nash equilibrium of this game.
\end{proof}

\section{Scatter Plot of Standard Deviation}
\label{appendixsection:stdev}
\begin{figure}[H]
\caption{Scatter plot of Standard Deviation}
\centering
\includegraphics[trim=0in 0in 0in 0in,scale=1.0]{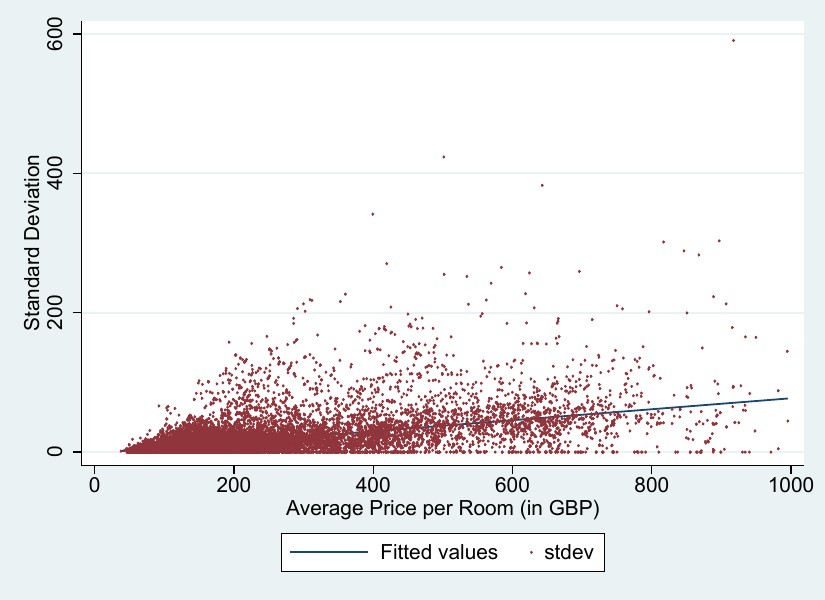}
%\floatfoot{
\newline
\tiny{Note: Here we have plotted the average prices for a \enquote{date of stay--date of booking--hotel id--room type} combination in the x-axis, and the standard deviation of prices in the y-axis.}
%}
\label{figure:stdev_london}
\end{figure}

In Figure \ref{figure:stdev_london}, we plot the standard deviation with respect to prices. As the figure shows, there is wide variation in the standard deviation values, suggesting the presence of price dispersion across websites.

\section{Robustness Checks}
\label{section:robustness_appendix}
\subsection{Dropping Single Website Observations}
Note that if prices are posted just on one website, then CV is equal to zero by definition. To make sure that such observations do not drive our results, we exclude those observations with a single website and run the regressions with the smaller sample. The results are presented in table \ref{cvreg_london_nosinglewebsite}. Comparing the results with table \ref{cvreg_london} suggests that our key results remain unaffected when we drop these observations.

\begin{table}[H]
\begin{center}
\caption{Regression results for CV after dropping observations with single website}
\label{cvreg_london_nosinglewebsite}
\scalebox{0.8}{
\begin{tabular}{lcccccc} \hline
 & (1) & (2) & (3) & (4) & (5) & (6) \\
Dependent variable: CV & Model1 & Model2 & Model3 & Model4 & Model5 & Model6 \\ \hline
 &  &  &  &  &  &  \\
Previous day CV & 0.767*** & 0.758*** & 0.758*** & 0.757*** & 0.756*** & 0.507*** \\
 & (0.00597) & (0.00598) & (0.00600) & (0.00601) & (0.00601) & (0.00845) \\
Number of Websites &  & 0.0302*** & 0.0305*** & 0.0311*** & 0.0316*** & 0.0669*** \\
 &  & (0.00253) & (0.00261) & (0.00264) & (0.00266) & (0.00396) \\
Page No in Skyscanner  &  & -0.000826*** & -0.000789*** & -0.000742*** & -0.000776*** & -0.00241*** \\
 &  & (0.000249) & (0.000255) & (0.000257) & (0.000258) & (0.000409) \\
Dummy for last 3 days &  & 0.00482** & 0.00478** & 0.00483** & 0.00499** & 0.00539*** \\
 &  & (0.00201) & (0.00201) & (0.00201) & (0.00202) & (0.00197) \\
No of Hotel Reviews &  &  & -0.000431 & -0.000229 & -0.000288 & 0.00459 \\
 &  &  & (0.00103) & (0.00105) & (0.00106) & (0.00940) \\
Hotel star rating &  &  & -0.0118 & -0.0162 & -0.0157 & 0.0269 \\
 &  &  & (0.0108) & (0.0111) & (0.0111) & (0.358) \\
Hotel Review Rating &  &  & -0.00465 & -0.0104 & -0.00776 & -0.0180 \\
 &  &  & (0.0129) & (0.0329) & (0.0330) & (0.213) \\
Constant & 0.0171*** & -0.0372*** & -0.00567 & -0.00455 & -0.0118 & -0.143 \\
 & (0.00102) & (0.00549) & (0.0276) & (0.0596) & (0.0601) & (0.592) \\
  &  &  &  &  &  &  \\
 &  &  &  &  &  &  \\ \hline
 &  &  &  &  &  &  \\ 
Observations & 10,817 & 10,817 & 10,817 & 10,817 & 10,817 & 10,817 \\
R-squared & 0.604 & 0.610 & 0.610 & 0.610 & 0.610 & 0.667 \\
Hotel FE & No & No & No & No & No & Yes \\
Date of Stay FE & No & No & No & No & Yes & Yes \\
 Quality FE & No & No & No & Yes & Yes & Yes \\ \hline
\multicolumn{7}{c}{ Standard errors in parentheses} \\
\multicolumn{7}{c}{ *** p$<$0.01, ** p$<$0.05, * p$<$0.1} \\ \hline
\end{tabular}
}
\end{center}
\end{table}

\subsection{Regressions with $\ln(CV)$ and $\ln(Range)$}
In addition to the results presented in Section \ref{section:priceconvergence}, here we present estimation results from two different specifications:
\begin{equation}
\ln (CV_t) = \alpha \ln (CV_{t-1}) + \beta_w X_w + \beta_h X_h + \epsilon       
\end{equation}
and 
\begin{equation}
    \ln (Range_t) = \alpha \ln (Range_{t-1}) + \beta_w X_w + \beta_h X_h + \epsilon
\end{equation}
where the range denotes the difference between the highest and lowest hotel room prices listed across different websites for a given date of stay and date of booking.  
The results documented in tables \ref{lncvreg_london} and \ref{lnrange_london} show that the estimated $\alpha$ is statistically significant and assumes a value less than 1, irrespective of the model specification. This is consistent with the results from Section \ref{section:priceconvergence} that compared to the previous day, on average, CV goes down the next day as we approach the date of stay.
\begin{table}[http]
\begin{center}
\caption{Regression results for $\ln$ (Coefficient of Variation)}
\label{lncvreg_london}
\scalebox{0.8}{
\begin{tabular}{lcccccc} \hline
 & (1) & (2) & (3) & (4) & (5) & (6) \\
Dependent variable: log(CV) & Model1 & Model2 & Model3 & Model4 & Model5 & Model6 \\ \hline
 &  &  &  &  &  &  \\
Previous day CV (in log) & 0.791*** & 0.775*** & 0.775*** & 0.774*** & 0.773*** & 0.527*** \\
 & (0.00590) & (0.00593) & (0.00595) & (0.00596) & (0.00597) & (0.00842) \\
Number of Websites &  & 0.0229*** & 0.0231*** & 0.0235*** & 0.0239*** & 0.0471*** \\
 &  & (0.00152) & (0.00157) & (0.00158) & (0.00159) & (0.00229) \\
Page No in Skyscanner &  & -0.000632*** & -0.000603*** & -0.000568*** & -0.000589*** & -0.00176*** \\
 &  & (0.000151) & (0.000155) & (0.000156) & (0.000156) & (0.000248) \\
Dummy for last 3 days &  & 0.00356*** & 0.00353*** & 0.00357*** & 0.00367*** & 0.00362*** \\
 &  & (0.00122) & (0.00122) & (0.00122) & (0.00123) & (0.00119) \\
No of Hotel Reviews &  &  & -0.000329 & -0.000163 & -0.000200 & 0.00323 \\
 &  &  & (0.000626) & (0.000636) & (0.000642) & (0.00569) \\
Hotel star rating &  &  & -0.00867 & -0.0117* & -0.0114* & 0.0109 \\
 &  &  & (0.00657) & (0.00672) & (0.00673) & (0.217) \\
Hotel Review Rating &  &  & -0.00418 & -0.000174 & 0.00124 & -0.0116 \\
 &  &  & (0.00784) & (0.0200) & (0.0200) & (0.129) \\
Constant & 0.0136*** & -0.0234*** & 0.00146 & -0.0109 & -0.0148 & -0.0793 \\
 & (0.000686) & (0.00326) & (0.0167) & (0.0362) & (0.0365) & (0.358) \\
 &  &  &  &  &  &  \\
 &  &  &  &  &  &  \\ \hline
 &  &  &  &  &  &  \\ 
Observations & 10,829 & 10,829 & 10,829 & 10,829 & 10,829 & 10,829 \\
R-squared & 0.624 & 0.632 & 0.632 & 0.632 & 0.633 & 0.688 \\
Hotel FE & No & No & No & No & No & Yes \\
Date of Stay FE & No & No & No & No & Yes & Yes \\
 Quality FE & No & No & No & Yes & Yes & Yes \\ \hline
\multicolumn{7}{c}{ Standard errors in parentheses} \\
\multicolumn{7}{c}{ *** p$<$0.01, ** p$<$0.05, * p$<$0.1} \\ \hline
\end{tabular}
}
\end{center}
\end{table}

\begin{table}[http]
\begin{center}
\caption{Regression results for $\ln$(Range). (Range refers to the gap between the highest and lowest prices posted for a hotel room across websites)}
\label{lnrange_london}
\scalebox{0.8}{
\begin{tabular}{lcccccc} \hline
 & (1) & (2) & (3) & (4) & (5) & (6) \\
Dependent variable: log(Range) & Model1 & Model2 & Model3 & Model4 & Model5 & Model6 \\ \hline
 &  &  &  &  &  &  \\
Previous day Range & 0.761*** & 0.667*** & 0.657*** & 0.656*** & 0.650*** & 0.324*** \\
 & (0.00635) & (0.00656) & (0.00669) & (0.00669) & (0.00676) & (0.00900) \\
Number of Websites &  & 1.018*** & 1.048*** & 1.061*** & 1.062*** & 1.591*** \\
 &  & (0.0285) & (0.0291) & (0.0294) & (0.0294) & (0.0381) \\
Page No in Skyscanner &  & 0.000947 & -0.00311 & -0.00248 & -0.00142 & -0.000361 \\
 &  & (0.00263) & (0.00269) & (0.00270) & (0.00271) & (0.00410) \\
Dummy for last 3 days &  & 0.0904*** & 0.0958*** & 0.0975*** & 0.0878*** & 0.0689*** \\
 &  & (0.0213) & (0.0212) & (0.0212) & (0.0213) & (0.0197) \\
No of Hotel Reviews &  &  & -0.0267** & -0.0224** & -0.0142 & 0.0229 \\
 &  &  & (0.0109) & (0.0110) & (0.0111) & (0.0940) \\
Hotel star rating &  &  & 0.619*** & 0.543*** & 0.539*** & -0.152 \\
 &  &  & (0.116) & (0.118) & (0.118) & (3.583) \\
Hotel Review Rating &  &  & 0.220 & 0.0481 & 0.157 & 2.089 \\
 &  &  & (0.136) & (0.347) & (0.347) & (2.127) \\
Constant & 0.683*** & -0.999*** & -2.291*** & -2.157*** & -2.519*** & -5.671 \\
 & (0.0216) & (0.0568) & (0.293) & (0.629) & (0.634) & (5.924) \\
 &  &  &  &  &  &  \\ \hline
 &  &  &  &  &  &  \\ 
Observations & 10,829 & 10,829 & 10,829 & 10,829 & 10,829 & 10,829 \\
R-squared & 0.570 & 0.616 & 0.618 & 0.619 & 0.620 & 0.707 \\
Hotel FE & No & No & No & No & No & Yes \\
Date of Stay FE & No & No & No & No & Yes & Yes \\
 Quality FE & No & No & No & Yes & Yes & Yes \\ \hline
\multicolumn{7}{c}{ Standard errors in parentheses} \\
\multicolumn{7}{c}{ *** p$<$0.01, ** p$<$0.05, * p$<$0.1} \\  \hline
\end{tabular}
}
\end{center}
\end{table}
\end{document}